\documentclass{article}
\usepackage{spconf}

\usepackage[english]{babel}
\usepackage{ifpdf}
\usepackage{cite} % Orders citations.
\usepackage{url}
\usepackage{hyperref}
\ifpdf
	\usepackage[pdftex]{graphicx}
	\graphicspath{{./figures/}}
\else
	\usepackage[dvips]{graphicx}
	\graphicspath{./figures/}
\fi
\usepackage{color}
\usepackage{pgf, tikz, pgfplots}
\usetikzlibrary{shapes, arrows, automata}
\usepackage{caption} % Extended options for captions.
\usepackage{subcaption} % Similar to caption but for subfigures.
\usepackage{amsmath}
\usepackage{amsfonts, amssymb, amsthm}
\usepackage{mathrsfs}
\usepackage{upgreek}
\usepackage{algorithm,algpseudocode}
\usepackage{enumerate}
\usepackage{multirow}
\usepackage{rotating}
\usepackage{needspace}

\input{mysymbol.sty}

\def\Tr{\mathsf{T}}

\newtheorem{lemma}{\hspace{0pt}\bf Lemma}
\newtheorem{proposition}{\hspace{0pt}\bf Proposition}

\newtheorem{theorem}{\hspace{0pt}\bf Theorem}

\newtheorem{definition}{\hspace{0pt}\bf Definition}

\newcommand{\QED}{\hfill\ensuremath{\blacksquare}}

\title{THE GRAPHON FOURIER TRANSFORM}
\name{Luana Ruiz$^{\dag}$, Luiz F. O. Chamon$^{\dag}$ and Alejandro Ribeiro$^{\dag}$\thanks{Supported by USA NSF CCF 1717120 and ARO W911NF1710438. 
		}
}
\address{\dag \ Department of Electrical and Systems Engineering, University of Pennsylvania, Philadelphia, USA
}
\begin{document}
\ninept
\maketitle
\begin{abstract}
In many network problems, graphs may change by the addition of nodes, or the same problem may need to be solved in multiple similar graphs. This generates inefficiency, as analyses and systems that are not transferable have to be redesigned. 
To address this, we consider graphons, which are both limit objects of convergent graph sequences and random graph models. We define graphon signals and introduce the Graphon Fourier Transform (WFT), to which the Graph Fourier Transform (GFT) is shown to converge. This result is demonstrated in two numerical experiments where, as expected, the GFT converges, hinting to the possibility of centralizing analysis and design on graphons to leverage transferability.
\end{abstract}
\section{Introduction}
\label{sec:intro}
System \textit{transferability} is a prevalent problem in signal processing, statistics and machine learning \cite{eisen2019optimal,elvin19-spectral}. For instance, when designing information processing architectures on networks that are bound to grow, we want to avoid making adjustments every time a new node is added to the network. This is the case, for instance, of streaming services that get thousands of new users every day and whose recommendation algorithms run on user-similarity networks \cite{weiyu18-movie,ruiz19-inv}.
Another example is reproducing a certain type of low-dimensional feature analysis on multiple instances of the same type of graph, eg., quantifying air pollution dispersion spectra on air quality sensor networks in different cities (cf. Section \ref{sec:sims}); it can be inefficient to have to recalculate the parameters of the transformation for each network where the analysis is repeated.

There are many reasons why redesigning a system or readjusting a transformation's parameters on different networks can be challenging. To begin with, in graphs whose number of nodes can grow, it becomes increasingly hard to measure the graph: it is either too big to obtain full measurements, or changes too frequently for any isolated measurement to be meaningful. This is a significant drawback in graph signal processing \cite{shuman13-mag,sandryhaila13-dspg,sandryhaila14-freq,ortega2018graph}, where problems like graph filter design usually assume full knowledge of the underlying graph \cite{segarra17-linear, sandryhaila2013discrete}. Very large graphs are also difficult to visualize \cite{wills1999nicheworks} and store, especially when they are not sparse. Even though this can be mitigated by dimensionality reduction techniques based on the graph's spectral decomposition \cite{rui2016dimensionality}, computing interior eigenvalues of large matrices is itself a costly operation \cite{morgan1991computing,paige1971computation}. Meanwhile, problems that involve replicating an analysis or design to a set of graphs with common characteristics suffer from similar issues: when there are many of these graphs, measuring, storing, visualizing and decomposing them can quickly become an expensive task. 

To avoid these issues, we propose a signal processing framework that enables a centralized and transferable approach to network systems based on \textit{graphons}. Graphons are infinite-dimensional representations of graphs that are at once limit objects of convergent graph sequences and random graph models  \cite{lovasz2012large}. The former makes them good representations of large graphs; the latter provides a probability model to which a family of graphs can be tied. 
The use of graphons is justified by the fact that, in practical problems involving graphs that grow or families of graphs, it is often the case that a certain large scale structure is retained (eg., the edges are always connected according to the same set of rules). This is exemplified in Figure \ref{fig:geom_graphs}, where the structural similarities between different geometric graphs of same size (left and center) and of different sizes (center and right) can be seen by noticing that there only exist edges between nodes positioned within a fixed radius of one another.
Graphons have been studied in multiple fields, and are the object of works that include accounting for uncertainty in the estimation of graph models \cite{wolfe2013nonparametric, airoldi2013stochastic} and computing properties such as node centrality \cite{avella2018centrality}, clustering coefficients \cite{diao2016model} and network game equilibria \cite{parise2019graphon} in very large networks.
In this paper, we lay the ground for a graphon signal processing framework by defining graphon signals as the limit objects of sequences of graph signals and by introducing the main building block of the framework --- the Graphon Fourier Transform --- in Def. \ref{defn:wft}. We also show that, for sequences of graphs converging to graphons, the Graph Fourier Transform converges to the Graphon Fourier Transform (Theorem \ref{thm:wft}), as suggested by the analysis done in \cite{morency2017signal} for some specific random graph models. This is an important result, as it hints to the possibility of centralizing the analysis and design of network systems on the graphon associated with a sequence or family of graphs.

\begin{figure*}[t]
\hspace{0.5cm}
\begin{minipage}[t]{0.33\linewidth}
  \centering
  \includegraphics[width=0.85\columnwidth, height = 0.45\columnwidth]{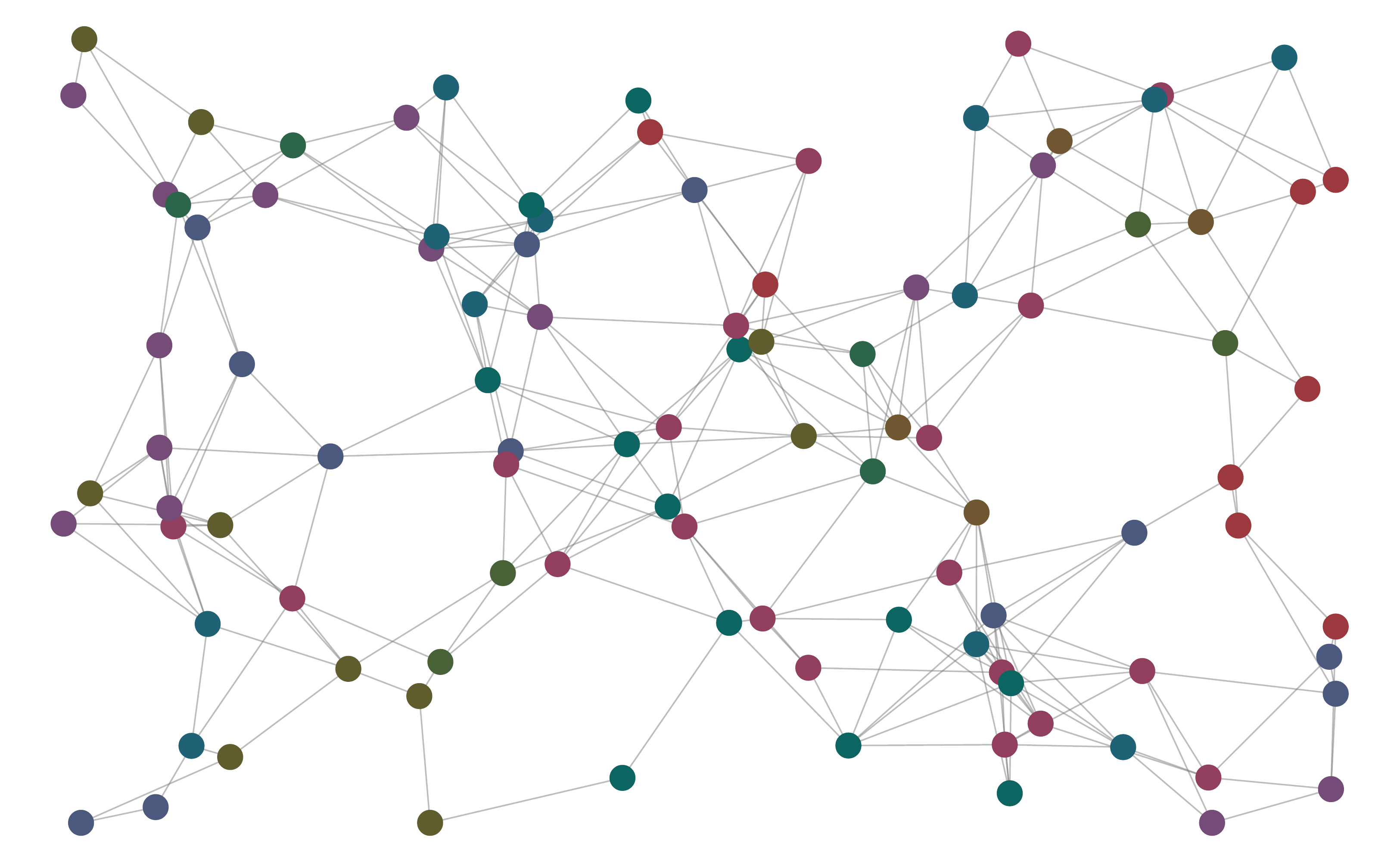}
  %\vspace{2.0cm}
  %\subcaption{$N=50$}
  \label{50nodes1}
\end{minipage}
\hspace{-0.5cm}
\begin{minipage}[t]{0.33\linewidth}
  \centering
  \includegraphics[width=0.85\columnwidth, height = 0.45\columnwidth]{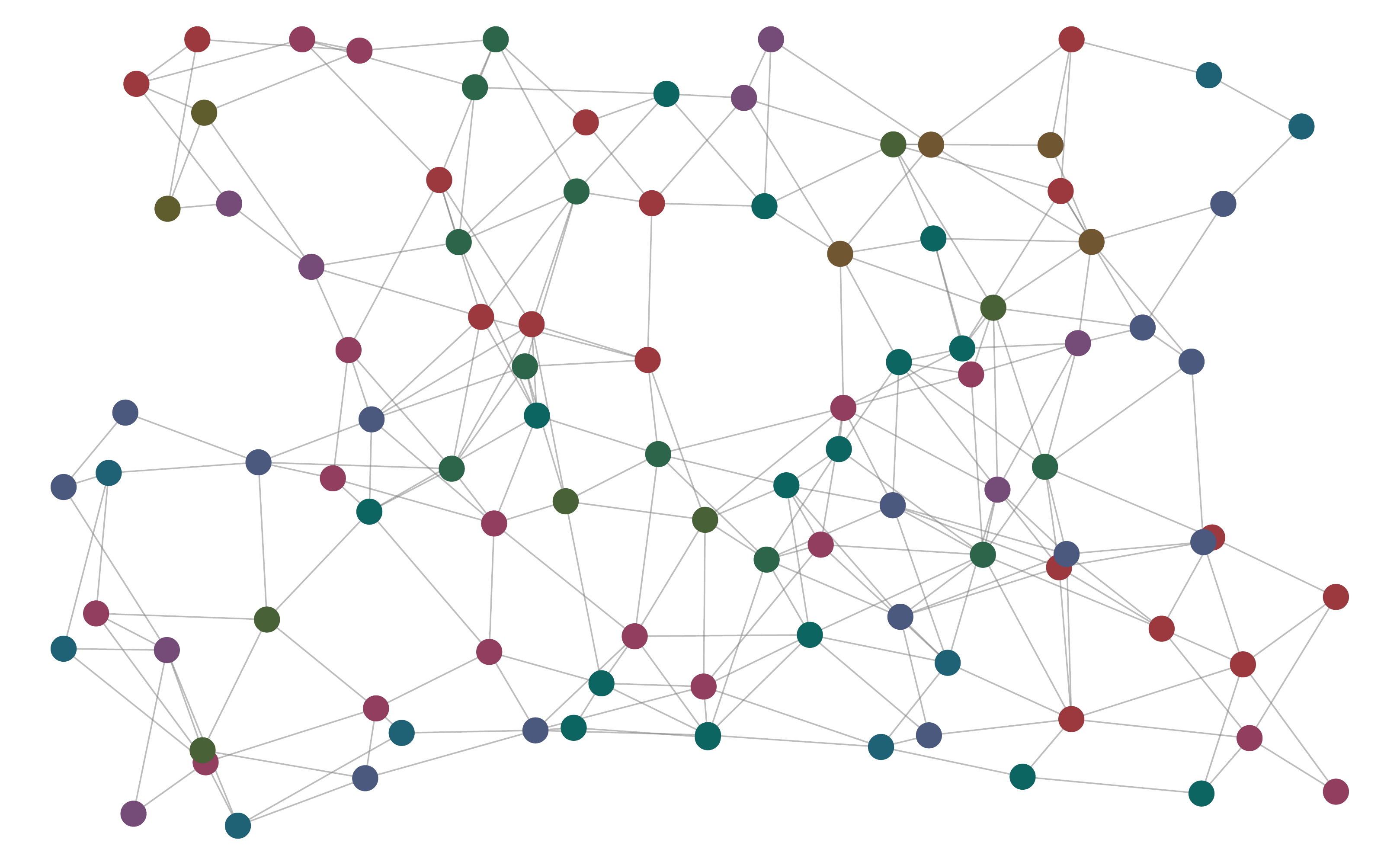}
  %\vspace{2.0cm}
  %\subcaption{$N=50$}
  \label{50nodes2}
\end{minipage}
\hspace{-0.5cm}
\begin{minipage}[t]{0.33\linewidth}
  \centering
  \includegraphics[width=0.85\columnwidth, height = 0.45\columnwidth]{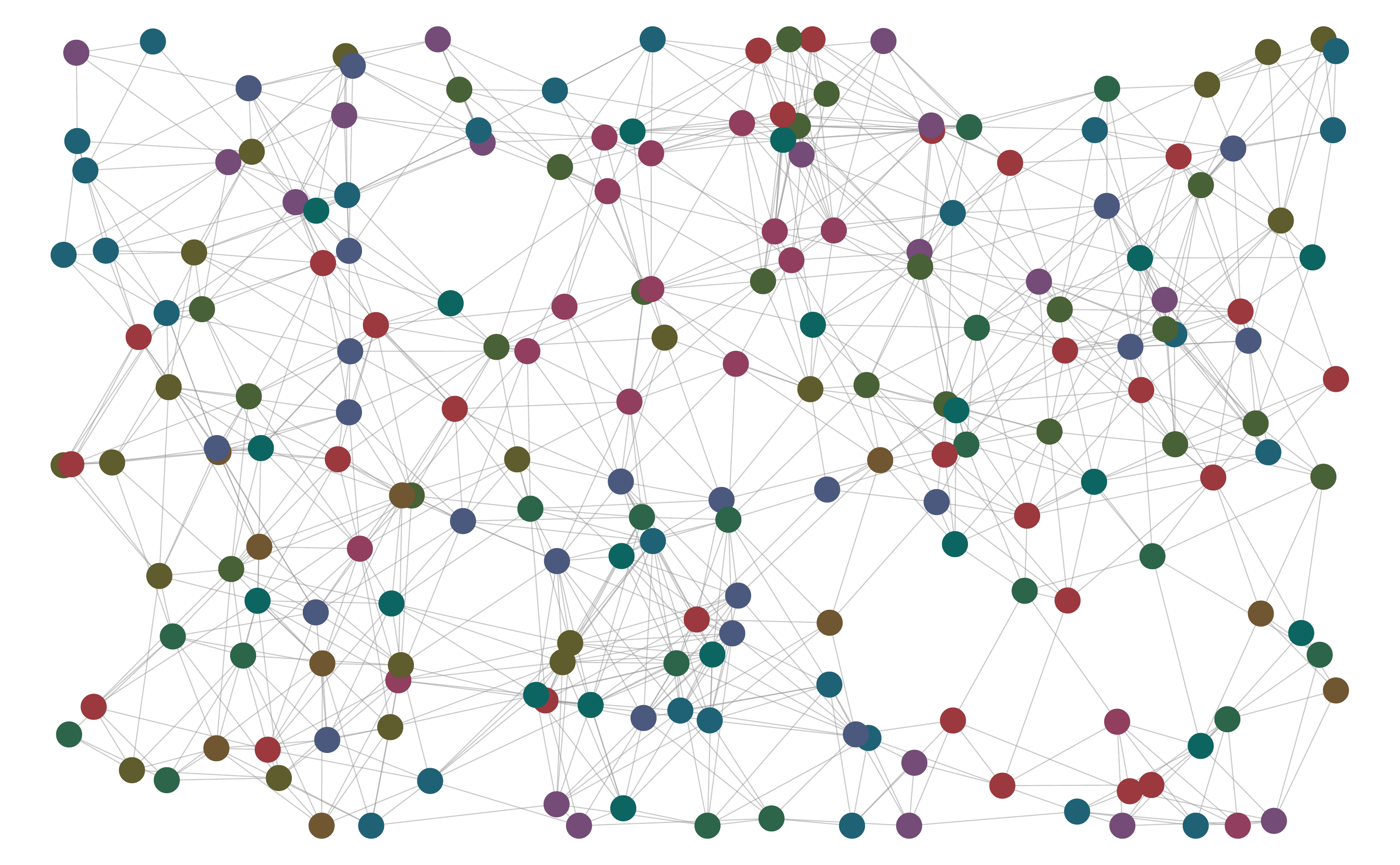}
  %\vspace{2.0cm}
  %\subcaption{$N=100$}
  \label{100nodes}
\end{minipage}
\caption{Examples of graphs with $n=50$ (left, center) and $n=100$ (right) drawn from a random geometric graph model with radius $0.25$.}
  \label{fig:geom_graphs}
\end{figure*}

In Section \ref{sec:sims}, Theorem \ref{thm:wft} is illustrated in two numerical experiments. In the first, we draw two $n$-node air pollution sensor networks from a common graphon model for growing $n$ and compare the GFTs of signals generated from the same pollution model. This experiment alludes to situations where we want to compare the spectra of the same type of signal on different graphs, in order to find similarities that are not noticeable on the vertex domain. This demonstrates the limit behavior described by Theorem \ref{thm:wft}, and, for large enough $n$, evidences the possibility of designing analyses and systems on one network (or on the graphon itself) and applying them to other graphs associated with the same graphon. In the second experiment, we compare the GFTs of movie ratings on small and large user networks to illustrate a more practical consequence of our results on graphs that, although inherently similar, are built from model-free data and are not related by any common generating graphon.

%Repeating the experiment for multiple stochastic realizations of the sensor networks, we observe that the norm difference of the GFT of the pollution signals measured on both graphs converges to 0 consistently as $n$ increases. 
%\red{Add exp. 2.}

%!TEX root = graphons-icassp19.tex

\section{Graphs and Graphons}
\label{sec:gsp_graphons}

Symmetric graphs are triplets $G = (V, E, A)$ where $V$ is a set of $n=|V|$ nodes, $E \subseteq V \times V$ is a set of edges and $A: E \to I \subseteq \reals$ is a weight function assigning real-valued weights $A(i,j)=A(j,i)$ in the set $I$ to edges $(i,j) \in E$. A graph is unweighted if $I = \{0,1\}$, in which case we omit $A$ and simply write $G = (V, E, A) = (V, E)$. Graphons are bounded, symmetric and measurable functions $\ccalW: [0,1]^2 \to [0,1]$. If we interpret points $u$ and $v$ of the unit line as nodes, $\ccalW(u,v)$ can be seen as the weight of the edge connecting $u$ and $v$. We can thus think of $\ccalW$ as a graph with an uncountable number of nodes that will serve as the limit of a sequence of graphs $\{G_n\}$ with growing number of nodes $n$.

To study convergence of graph sequences to graphons consider arbitrary unweighted and undirected reference graphs $F = (V', E')$. We define homomorphisms of $F$ into $G$ as adjacency preserving maps $\beta:V'\to V$ in which $(i,j)\in E'$ implies $(i,j)\in E$. There are a total of $|V|^{|V'|} = n^{n'}$ maps from $V'$ to $V$ but only some of them are adjacency preserving homomorphisms. We therefore define $\mbox{hom}(F,G) = \sum_{\beta} \prod_{(i,j) \in E'} A(\beta(i),\beta(j))$ as a weighted sum of the total number of homomorphisms that map $F$ into $G$ and the density of homomorphisms as the ratio \cite{lovasz2012large},
\begin{equation} \label{eqn:hom_density}
   t(F, G) := \frac{\mbox{hom}(F,G)}{n^{n'}} 
          := \frac{\sum_{\beta} \prod_{(i,j) \in E'} A(\beta(i),\beta(j))}{n^{n'}} .
\end{equation}
The interpretation of \eqref{eqn:hom_density} is easiest when $G$ is unweighted. In that case $\prod_{(i,j) \in E'} A(\beta(i),\beta(j))=1$ and we are just counting the total number of ways in which the reference graph $F$ can be mapped into $G$. We can think of $F$ as a motif and of $t(G,F)$ as how often that motif appears in $G$ relative to the maximum number of times the motif could appear.

The homomorphism density in \eqref{eqn:hom_density} can be generalized to graphons if we replace the sum by an integral. Thus, for a given reference motif $F$ we define the homomorphism density of $F$ into the graphon $\ccalW$ as
\begin{equation}\label{eqn_hom_density_graphon}
   t(F,\ccalW) := \int_{{[0,1]}^{V'}} 
                      \prod_{(i,j) \in E'} \ccalW(u_i,u_j) 
                            \prod_{i \in V'} du_i .
\end{equation}
Given eqs. \ref{eqn:hom_density} and \ref{eqn_hom_density_graphon} we say that the graph sequence $\{G_n\}$ converges to the graphon $\ccalW$ if for any given motif $F$ it holds
\begin{equation} \label{eqn_graphon_convergence}
   \lim_{n\to\infty} t(F,G_n) = t(F,\ccalW).
\end{equation}
To understand the meaning of \eqref{eqn_graphon_convergence}, it is instructive to consider graphs $G_n$ sampled from a graphon $\ccalW$, which are called $\ccalW$\textit{-random graphs}. These graphs have labels $ u_i \sim U [0,1]$ drawn uniformly and independently at random from $[0,1]$, and edge sets such that $(u_i, u_j) \in E$ with probability $\ccalW(u_i,u_j)$. It is possible to see that \eqref{eqn_graphon_convergence} holds with probability 1 for sequences of $\ccalW$-random graphs. The convergence mode in \eqref{eqn_graphon_convergence} allows for more general generative models asides from these randomly sampled graphs. We point out that every graphon is the limit object of a sequence of convergent graphs and, conversely, that every convergent graph sequence converges to a graphon \cite{lovasz2012large} --- i.e., if the limit $\lim_{n\to\infty} t(F,G_n)$ exists it can be written as an integral of the form of \eqref{eqn_hom_density_graphon}.

%!TEX root = graphons-icassp19.tex

\section{Graph Signals, Graphon Signals and their Fourier Transforms}
\label{sec:wft}

We define a graph signal as the pair $(G,\bbx)$ where $G$ is the graph on which the signal is supported and $\bbx \in \reals^n$ is such that $[\bbx]_i$ stands for the value of this signal at node $i$ of $G$. The vector $\bbx$ can be seen as a \textit{vertex representation} of data, but graph signals also admit a \textit{spectral representation} obtained by application of the Graph Fourier Transform (GFT) \cite{sandryhaila14-freq}. This is defined with respect to a shift operator $\bbS$, which is a stand-in for a matrix representation of the graph $G$. Since we are working with undirected graphs we have that $\bbS$ is symmetric and that it therefore accepts an eigenvector decomposition of the form $\bbS= \bbV\bbLam\bbV^H$ with $\bbV$ being the orthonormal matrix of eigenvectors and $\bbLam$ a diagonal matrix containing the eigenvalues of $\bbS$. The GFT of $(G,\bbx)$ is defined as
\begin{equation}
   \hbx_G = \mbox{GFT}\{(G,\bbx)\} := \bbV^\mathsf{H} \bbx
\end{equation}
which is equivalent to decomposing $(G,\bbx)$ in the spectral basis of $G$. Where there is no ambiguity, we may omit the subscript $G$. Because $\bbV^\mathsf{H}$ is orthonormal, an inverse transformation can also be defined. We write the inverse Graph Fourier Transform (iGFT) of $\hbx_G$ as $\mbox{iGFT}\{\hbx_G\} := \bbV \hbx_G = \bbx$.
The GFT allows switching from the vertex to the spectral domain, and the iGFT provides the way back. The goal of this section is to define graphon analogous of graph signals, GFTs, and iGFTs.

\subsection{Graphon Signals and Limits of Graph Signal Sequences} \label{sbs:graphon_signals}

A graphon signal is a pair $(\ccalW,X)$, where $\ccalW$ is a graphon and $X: [0,1] \to \reals$ is a function mapping points $u \in [0,1]$ onto real numbers $X(u)$. We only consider finite energy graphon signals $X \in L^2$. Graphon signals can be seen as the continuous counterpart of graph signals. 
It will again be important to consider $X$ in conjunction with $\mathcal{W}$.

Every graph signal $(G,\bbx)$ induces a graphon signal $(\ccalW_G, X_G)$. This ``continuous'' representation of a graph signal will be useful in the derivations carried out in Section \ref{sec:convergence}. Explicitly, let $I_j$, $j = 1,\dots,N$, be a partition of the unit interval. Then, $(\ccalW_G, X_G)$ is defined as
\begin{align}\label{E:induced_graphon}
	\ccalW_G(u,v) = [\bbS]_{jk} \times \mbI\left( u \in I_j \right)
		\mbI\left( v \in I_k \right)
	\\
	X_G(v) = [\bbx]_{k} \times \mbI\left( v \in I_k \right)
		\text{.}
\end{align}

On sequences of graphs $\{G_n\}$ converging to a graphon $\ccalW$, sequences of graph signals $\{(G_n,\bbx_n)\}$ converging to graphon signals $(\ccalW, X)$ can always be defined. We characterize their convergence in Def. \ref{defn:convergent_pairs}. 
%Conversely, notice that when generating W-random graphs $G_n$ from $\mathcal{W}$ --- by sampling nodes $u_1, \ldots, u_n$ uniformly at random on the unit interval --- we can also come up with graph signals $\bbx_n$ which are evaluations of $X$ at the sample points $u_i$. Formally, this can be expressed as $[\bbx_n]_i = X(u_i)$, and following this construction for $n \to \infty$, we have, in a slight abuse of notation, $\bbx \to X$ pointwise.
%As much as this is a useful convergence result, we want to be able to characterize convergence for arbitrary convergent sequences of graph signals, that is, for a broader class of signals than those generated through sampling nodes $u_i$ and evaluating a ``generative'' graphon signal $X(u_i)$. In order to do this, we define \textit{graph signal-graph} pairs $(\bbx, G)$ and \textit{graphon signal-graphon} pairs $(X, \mathcal{W})$, and characterize convergence of $(\bbx, G)$ to $(X, \mathcal{W})$ in the following Definition.
\begin{definition}[Convergent sequences of graph signals] \label{defn:convergent_pairs}
A sequence of graph signals $\{(G_n, \bbx_n)\}$ is said to converge to the graphon signal $(\mathcal{W}, X)$ if there exists a sequence of permutations $\{\pi_n\}$ such that
\begin{equation}
t(F,G_n) \to t(F,\mathcal{W})
\end{equation}
for any unweighted and undirected graph $F$ and
\begin{equation}
\|X_{\pi_n(G_n)}-X\|_{L^2} \to 0
\end{equation}
where $(\ccalW_{\pi_n(G_n)}, X_{\pi_n(G_n)})$ denotes the step function graphon signal induced by $(\pi_n(G_n),\pi_n(\bbx_n))$.
\end{definition}
In short, we say that $\{(\bbx_n, G_n)\}$ converges to $(X,\mathcal{W})$ when: (i) $G_n$ converges to $\mathcal{W}$ in the homomorphism density sense, and (ii) for some sequence of permutations $\{\pi_n\}$ of the labels of each $G_n$, the graphon signal induced by $\bbx_n$ converges to $X$ in $L^2$.
Notice that, similarly to how graphons model graph families, graphon signals can also be interpreted as models for network phenomena. A graph signal can be sampled from the graphon signal at the same locations $u_i$ that a $\ccalW$-random graph is sampled from the graphon. Since $\ccalW$-random graphs converge to the graphon in probability, as long as the graphon signal is in $L^2$ any sequence of graph signals generated in this way can be shown to converge to the graphon signal (in the sense of Def. \ref{defn:convergent_pairs}) with probability 1.

%\begin{remark} 
%It is important to point out that Definitions \ref{defn:linear} and \ref{defn:wft}, as well as the definition of a graphon signal, are not realizable in the way that graph signals, graph filters and the Graph Fourier Transform are. Unlike graphs, graphons are intangible objects, but their value lies in providing a complete and concise representation of multiple graphs belonging to the same ``class'' (eg. \cite{schaub19-spectral,wolfe2013nonparametric}), or of \textit{very large} graphs whose properties converge towards those of the graphon, as extensively demonstrated in \cite{lovasz2012large}. Similarly to \cite{avella2018centrality}, where the notion of centrality was extended to graphons, the definitions in this section should be interpreted as merely a parallel to graph signal processing concepts, intended to provide a comparison with the case of discrete graphs and to enable the convergence analysis carried out in Section \ref{sec:filters}.
%\end{remark}

%In this section, we define the Graphon Fourier Transform for graphon signals and show that, for certain graph signal-graph pairs $(G_n, \bbx_n)$ converging to a graphon signal-graphon pair $(X, \mathcal{W})$, the Fourier transform of $\bbx_n$ converges to the Fourier transform of $X$. 

%We start by analyzing the linear operator associated with $\mathcal{W}$. 

\subsection{Graphon Fourier Transform} \label{sbs:wft}

Because $\mathcal{W}$ is bounded, it defines a Hilbert-Schmidt integral operator $T_\mathcal{W}: L^2 \to L^2$ on $(\ccalW,X)$,
\begin{equation} \label{eqn:graphon_shift}
(T_\mathcal{W} X)(v) := \int_0^1 \mathcal{W}(u,v)X(u)du
\end{equation}
which we call the \textit{graphon shift operator} (WSO). Since $\ccalW$ is symmetric, $T_\ccalW$ is a self-adjoint operator that can be decomposed as $ 
\mathcal{W}(u,v) \sim \sum_{i=0}^{\infty} \sigma_i \varphi_i(u) \varphi_i(v)$, where $\sigma_i \in [-1,1]$ and $\varphi_i: [0,1] \to \mbR$ are countable. Splitting between positive and negative eigenvalues, we can reorder the $\sigma_i$ and $\varphi_i$ as $\sigma_j$ and $\varphi_j$ with $j \in \mbZ\setminus\{0\}$ such that $1 \geq \sigma_1 \geq \sigma_2 \geq \ldots \geq 0 \geq \ldots \geq \sigma_{-2} \geq \sigma_{-1} \geq -1$. In this ordering, $\sigma_j \to 0$ as $|j| \to \infty$ due to compactness of $T_\ccalW$; zero is the only possible point of accumulation, which means that all $\sigma_j \neq 0$ have finite multiplicities \cite[Chapter 28, Theorem 3]{lax02-functional}. 
The eigenfunctions $\varphi_j$ form the \textit{graphon spectral basis}, which is a complete orthonormal basis of $L^2[0,1]$. The inner products $\langle X, \varphi_j\rangle$ thus yield a complete representation of $(\ccalW,X)$ that we call the Graphon Fourier Transform.

\begin{definition}[Graphon Fourier Transform]  \label{defn:wft}
Consider the graphon signal $(\ccalW,X)$ with $\mathcal{W}: [0,1]^2 \to [0,1]$ and $X: [0,1] \to \mbR$. Let $\{\sigma_j\}$ and $\{\varphi_j\}$ form the spectral decomposition of $T_\mathcal{W}$. 
Then, the Graphon Fourier Transform (WFT) of $(\ccalW,X)$, denoted $\hat{X} = \mbox{WFT}\{(\ccalW,X)\}$, is
\begin{equation}
[\hat{X}]_j = \hat{X}(\sigma_j) = \langle X, \varphi_j\rangle = \int_{0}^1 X(u) \varphi_j(u) du
\end{equation}
and the inverse Graphon Fourier Transform (iWFT) of $\hat{X}$ is $
\mbox{iWFT}\{\hat{X}\} := \sum_j \hat{X}(\sigma_j) \varphi_j = X$.
\end{definition}

Existence of the iWFT is guaranteed by orthonormality of the $\{\varphi_j\}$. Using the WFT, we can define graphon signals that are \textit{bandlimited}.

\begin{definition}[Bandlimited graphon signals] \label{defn:bandlimited}
A graphon signal $(\ccalW,X)$ is said to be $c$-bandlimited if, for some $ c \in (0,1)$, the WFT coefficients $\hat{X}(\sigma_j)$ for which $j \in \{k\ \text{ s.t. }\ |\sigma_k| < c\}$  are zero. In particular, $\hat{X}$ has finite dimension for bandlimited graphon signals $(\ccalW, X)$.
\end{definition}

It is important to point out that Def. \ref{defn:wft}, as well as the definition of a graphon signal, are not realizable in the way that graph signals and the GFT are. Unlike graphs, graphons are intangible objects, but their value lies in providing a complete and concise representation of multiple graphs belonging to the same ``class'' (eg. \cite{schaub19-spectral,wolfe2013nonparametric}), or of \textit{very large} graphs whose properties converge towards those of the graphon, as extensively demonstrated in \cite{lovasz2012large}. Similarly to \cite{avella2018centrality}, where the notion of centrality was extended to graphons, the definitions in this section should be interpreted as merely a parallel to graph signal processing concepts, intended to provide a comparison with the case of discrete graphs and to enable the convergence analysis carried out in Section \ref{sec:convergence}.

%!TEX root = graphons-icassp19.tex

\section{The GFT converges to the WFT} \label{sec:convergence}

The WFT only depends on the graphon through its eigenfunctions, whereas the GFT only depends on the graph's eigenvectors. Hence, for a convergent sequence of graph signals (Def. \ref{defn:convergent_pairs}), convergence of the GFT to the WFT is equivalent to the graph's eigenvectors converging to the graphon's eigenfunctions. 
When graph signals converge to a graphon signal that is \textit{bandlimited}, and when the underlying graphon is \textit{non-derogatory}, this is indeed the case (cf. Lemma \ref{T:eigenvectors_convergence}). This is stated in Theorem \ref{thm:wft}, but first we define non-derogatory graphons in Def. \ref{defn:non-derog}. 
\begin{definition}[Non-derogatory graphons] \label{defn:non-derog}
A graphon $\ccalW$ is said to be non-derogatory if $\sigma_j \neq \sigma_k$ for $j \neq k$.
\end{definition}
\begin{theorem}[GFT convergence] \label{thm:wft}
Let $\{(G_n,\bbx_n)\}$ be a sequence of graph signals converging to the $c$-bandlimited graphon signal $(\ccalW,X)$, where $\ccalW$ is non-derogatory. Then, there exists a sequence of permutations $\pi_n$ such that $\text{GFT}\{(\pi_n(G_n),\pi_n(\bbx_n))\} \to \text{WFT}\{(\ccalW,X)\}$ in the sense that $[\hbx_n]_j/\sqrt{n} \to [\hat{X}]_j$ for all $j \in \mbZ \setminus \{0\}$ as $n \to \infty$.
\end{theorem}

This theorem has three main takeaways. 
First, it allows drawing insights about the general spectral behavior of a graph signal without knowing the graph and/or the signal exactly, so long as the ``class'' to which the graph belongs (its generating graphon) and the generative model of the graph signal (the graphon signal) are known.
Secondly, it links the Fourier transform on the ``discrete'' domain of graphs to the Fourier transform on the ``continuous'' domain of graphons. In the limit, this relationship resembles that of the discrete Fourier transform (DFT) with the continuous Fourier transform (CFT) for time signals. Finally, Theorem \ref{thm:wft} also works in the opposite direction, meaning that if we sample a graph signal from the graphon signal, its GFT can be sampled from the WFT. This follows from the fact that sampled sequences of graph signals converge to the generating graphon signal with probability 1. We also point out that the requirement that the graphon be non-derogatory is not very restrictive because the space of non-derogatory graphons is dense in the space of graphons (cf. Prop. \ref{T:non_derog_graphons_dense} in the appendices), and so, for every derogatory graphon, there exists a non-derogatory graphon arbitrarily close (in the appropriate metric) for which Theorem \ref{thm:wft} holds. 

We now proceed to proving Theorem \ref{thm:wft}.

\begin{proof}[Proof of Theorem \ref{thm:wft}]
Although the GSOs $\bbS_n$ of the graphs $G_n$ have a finite number of eigenvalues $\lambda_j$, we maintain the convention of associating the eigenvalue sign with its index and ordering the eigenvalues in decreasing order of absolute value. Hence, the indices $j$ are now defined on some finite set $L \subseteq \mbZ \setminus \{0\}$.
It will be useful to consider the continuous representation of the graph signals $(G_n, \bbx_n)$ provided by their induced graphon signals $(\ccalW_{G_n},X_{G_n})$ (eq. \eqref{E:induced_graphon}). It is fundamental that $(\ccalW_{G_n},X_{G_n})$ retains the spectral properties of $(G_n, \bbx_n)$, which is guaranteed by Lemma \ref{T:induced_graphon} below.
\begin{lemma}\label{T:induced_graphon}
Let $(\ccalW_G,X_G)$ be the graphon signal induced by the $n$-node graph signal $(G,\bbx)$ as defined in \eqref{E:induced_graphon}. Then, for $j \in L$ we have
\begin{align*}
	\sigma_j(T_{\ccalW_G}) &= \frac{\lambda_{j}(\bbS)}{n}
	\\
	\varphi_j(T_{\ccalW_G})(u) &= [\bbv_j]_{k}
		\times \sqrt{n} \mbI\left( u \in I_k \right)
	\\
	[\hat{X}_G]_j &= \dfrac{[\hbx]_j}{\sqrt{n}}
\end{align*}
For $j \notin L$, $\sigma_j(T_{\ccalW_G}) = [\hat{X}_G]_j = 0$ and $\varphi_j(T_{\ccalW_G}) = \phi_j$, such that $\{\varphi_j(T_{\ccalW_G})\} \cup \{\phi_j\}$ forms an orthonormal basis of $L^2([0,1])$.
\end{lemma}
\begin{proof} \renewcommand{\qedsymbol}{}
Refer to the appendices.
\end{proof}
It then suffices to show $\text{WFT}(\ccalW_{\pi_n(G_n)},\pi_n(X_{G_n})) \to \text{WFT}(\ccalW,X)$, since $L \to \mbZ \setminus \{0\}$ as $n \to \infty$. We will leave the dependence on $\pi_n(G_n)$ implicit and write $\ccalW_n = \ccalW_{\pi_n(G_n)}$ and $X_n = \pi_n(X_{G_n})$. We proceed by using the following lemma, whose proof can be found in the appendices. Our result then follows from the fact that inner products are continuous in the product topology they induce.

\begin{lemma}\label{T:eigenvectors_convergence}
Let $C = \{j \in \mbZ \setminus \{0\} \mid |\sigma_j(T_{\ccalW})| \geq c\}$ be the set of indices of the non-vanishing eigenvalues and $\ccalE$ denote the subspace spanned by the eigenfunctions $\{\varphi_j(T_\ccalW)\}_{j \notin C}$. Then, $\varphi_j(T_{\ccalW_n}) \to \varphi_j(T_\ccalW)$ weakly for $j \in C$ and $\varphi_j(T_{\ccalW_n}) \to \phi \in \ccalE$ for $i \notin C$.
\end{lemma}
\begin{proof} \renewcommand{\qedsymbol}{}
Refer to the appendices.
\end{proof}
Start by the eigenvectors with indices in $C$. For any $\epsilon > 0$, it holds from Lemma \ref{T:eigenvectors_convergence} and from $X_n \to X$ in $L^2$, that there $\exists\ N_1, N_2$ such that
\begin{align*}
	\|\varphi_j(T_{\ccalW_n}) - \varphi_j(T_\ccalW)\|_{L^2} \leq \frac{\epsilon}{2\|X\|}
		\text{, for all } n > N_1
	\\ \text{and} \quad
	\|X_n - X\|_{L^2} \leq \frac{\epsilon}{2}
		\text{, for all } n > N_2
\end{align*}
where $\|\varphi_j(T_{\ccalW_n})\| \leq 1$ for all $n$ and $j \in C$ because the graphon spectral basis is orthonormal. Since the sequence $\{X_n\}$ is convergent, it is bounded and $\|X\| < \infty$. Let $M = \max(N_1,N_2)$. Then, it holds that
\begin{align*}
	|[\hat{X}_{n}]_j - [\hat{X}]_j|
	{}&\leq \frac{\epsilon}{2} \|\varphi_j(T_{\ccalW_n})\| + \|X\|\frac{\epsilon}{2\|X\|} \leq \epsilon
\end{align*}
%
%		&= |\langle{X_n},{\varphi_j(T_{\ccalW_n})}\rangle - \langle{X},{\varphi_j(T_\ccalW)}\rangle|
%	\\
%	{}&= |\langle{X_n - X},{\varphi_j(T_{\ccalW_n})}\rangle + \langle{X},{\varphi_j(T_{\ccalW_n}) - \varphi_i(T_\ccalW)}\rangle|
%	\\
%	{}&\leq \|X_n - X\| \|\varphi_j(T_{\ccalW_n})\| + \|X\|\|\varphi_j(T_{\ccalW_n}) - \varphi_j(T_\ccalW)\|
%	\\
for all $n > M$, which we obtain by expressing $[\hat{X}_{n}]_j$ and $[\hat{X}]_j$ as inner products and applying the Cauchy-Schwarz inequality once.
The coefficients with indices in $C$ of the GFT of $(G_n,X_n)$ therefore converge to those of the WFT of $(\ccalW,X)$ as $n \to \infty$.

For $j \notin C$, the eigenfunctions $\varphi_j(T_{\ccalW_n})$ may not converge to $\varphi_j(\ccalW)$, but they do converge to some function $\phi \in \ccalE$. Since the graphon signal $(\ccalW,X)$ is $c$-bandlimited, we have $\langle{X},{\varphi_j(T_\ccalW)}\rangle = 0$ for all $j \notin C$.
%, so that $X$ must be orthogonal to all functions in $\ccalE$. 
Using the same argument as above yields that the remaining GFT coefficients also converge to the WFT. %Explicitly,
%
%\begin{equation*}
%	\langle{\varphi_j(T_{\ccalW_n})},{X_n}\rangle \to \langle{\phi},{X}\rangle = 0 = \langle{\varphi_j(T_\ccalW)},{X_n}\rangle
%		\text{.}
%\end{equation*}
%
\end{proof}

\begin{figure}[t]
\begin{minipage}[t]{1.0\linewidth}
  \centering
  \includegraphics[width=0.9\columnwidth, height=0.55\columnwidth]{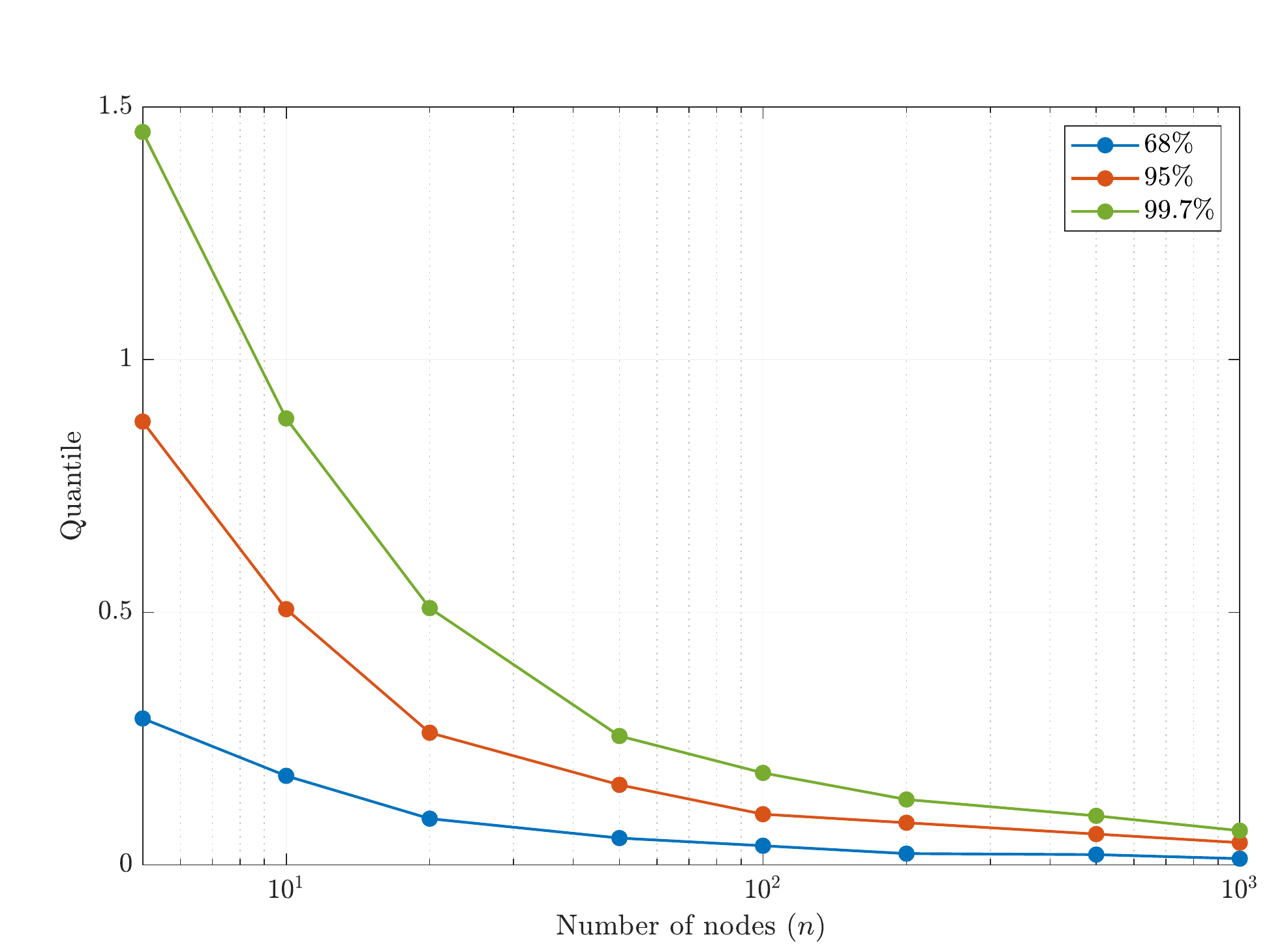}
  %\vspace{2.0cm}
\end{minipage}
\caption{Quantiles ($68\%, 95\%, 99.7\%$) of the minimum normalized difference between GFTs of air pollution signals on $G_1$ and $G_2$ over 50 iterations for increasing $n$.}
  \label{fig:sensor_net}
\end{figure}

\section{Numerical Experiments}
\label{sec:sims}

\subsection{Convergence of air pollution signal spectra}

We compare the spectral representations of air pollution signals measured on two distinct sensor networks of same size and drawn from the same random graph model, and show that the GFTs of the pollution signals on both networks converge as the number of nodes $n$ increases. This problem can be seen as comparing the spectra of air pollution signals in two cities.  
Two air pollution sensor networks are drawn from a soft random geometric graph model \cite{penrose2016connectivity} where, given nodes $i,j$ and their $(x,y)$ coordinates, the probability of the edge $(i,j)$ is
\begin{equation} \label{eqn:srgg}
p(i,j) \propto \exp \big(-\beta \sqrt{(x_i-x_j)^2 + (y_i-y_j)^2}\big)\ .
\end{equation}
Assuming that we only care about pollutant dispersion in the cross-wind direction $y \geq 0$, the $x$ coordinate is fixed so that eq. \eqref{eqn:srgg} only depends on $y_i$ and $y_j$. We also normalize $y_i$ and $y_j$ by setting $u_i=y_i/y_{\mbox{\tiny max}}$ and $u_j=y_j/y_{\mbox{\tiny max}}$ so that eq. \eqref{eqn:srgg} fits the expression of a graphon. 

The pollution dispersion model is given by $S(y) \propto \exp (-{y^2}/{2\sigma_y^2})$, where we have fixed the source of pollution at $y=0$ and $\sigma_y^2$ is a cross-wind mixing term \cite{arya1999air}. 
By normalizing $y$ like before as $u = y/y_{\mbox{\tiny max}}$, $S(u)$ can be interpreted as a signal on the graphon of eq. \ref{eqn:srgg}. 

For a range of values of $n$, we sample two $n$-node graphs $G_1$ and $G_2$ at locations $\{u^1_i\}$ and $\{u^2_i\}$ drawn i.i.d from $U[0,1]$. Then, we generate graph signals $(G_1,\bbs_1)$ and $(G_2,\bbs_2)$ by evaluating $[\bbs_1]_i = S(u^1_i)$ and $[\bbs_2]_i = S(u^2_i)$ for $1 \leq i \leq n$.
Because of the way in which $G_1$ and $G_2$ and $\bbs_1$ and $\bbs_2$ are sampled, following Theorem \ref{thm:wft} the GFTs of $(G_1,\bbs_1)$ and $(G_2,\bbs_2)$ should both converge to the WFT of $S$ with probability 1. We verify this empirically by computing the GFTs of $(G_1,\bbs_1)$ and $(G_2,\bbs_2)$ and comparing their minimum norm difference.
%on graphs drawn from the same model ($\bbG_1$ and $\bbG_2$) and from different models ($\bbG_2$ and $\bbG_{\mbox{{\tiny ER}}}$). 

The $68\%$, $95\%$ and $99.7\%$ quantile curves of the GFT norm difference for $50$ realizations of this experiment are plotted in Figure \ref{fig:sensor_net}. The minimum norm difference was computed by sorting and subtracting $\hbs_1$ and $\hbs_2$, and the experiment was run for graphs of increasing size $n$. Figure \ref{fig:sensor_net} shows that the GFT converges as expected: across realizations, the minimum norm difference of the GFTs gets more and more concentrated around $0$ as $n$ increases.

\subsection{GFT transferability in recommender systems}

In this experiment, we use the MovieLens dataset \cite{harper16-movielens} with 100,000 ratings from 943 users to 1,682 movies to analyze the behavior of the GFT on user similarity networks of increasing size. The networks are built by computing Pearson correlations between ratings that users have given to movies, in the fashion of \cite{weiyu18-movie}. The ratings vary between 1 and 5 and are seen as signals on the user network. We consider the movie ``Toy Story'', and use as the reference graph signal $\bbx$ the user ratings predicted for this movie on the full 943-user network and using the method in \cite{weiyu18-movie}. 

%To illustrate an application of Theorem \ref{thm:wft} in a real-world setting where transferability could be leveraged, 
The experiment is as follows: we (i) generate a graph from a pool of users of size $n$, (ii) compute its spectral basis, (iii) sample the graph signal $\bbx_n$ from $\bbx$ and (iv) compare the GFT of $\bbx_n$ in the $n$-node user network, $[\hbx_n]_{i=1}^n$, with the GFT coefficients $[\hbx]_{i=1}^n$  of the complete signal $\bbx$ associated with the $n$ largest eigenvalues of the full 943-user network. Although these are networks built from real data to which we cannot attribute a generating graphon, our goal here is to illustrate how our results can be implicitly observed and used to leverage transferability even in graphs that are not related by a common statistical model, but that are ``similar'' in some empirical sense.

The average relative norm difference between the GFTs and its standard deviation are shown in Figure \ref{fig:movie_coeffs} for 10 realizations of networks with $n$ users each. We see that the GFT difference consistently decreases with $n$. Perhaps more importantly, notice that this difference is as low as $0.5\%$ on average for networks as small as $n=200$. In other words, the spectral decomposition of the signal corresponding to ``Toy Story'' ratings in a network with less than a quarter of the users of the full network is already an accurate enough representation of the complete spectra of this signal.

\begin{figure}[t]
\begin{minipage}[t]{1.0\linewidth}
  \centering
  \includegraphics[width=0.9\columnwidth, height=0.55\columnwidth]{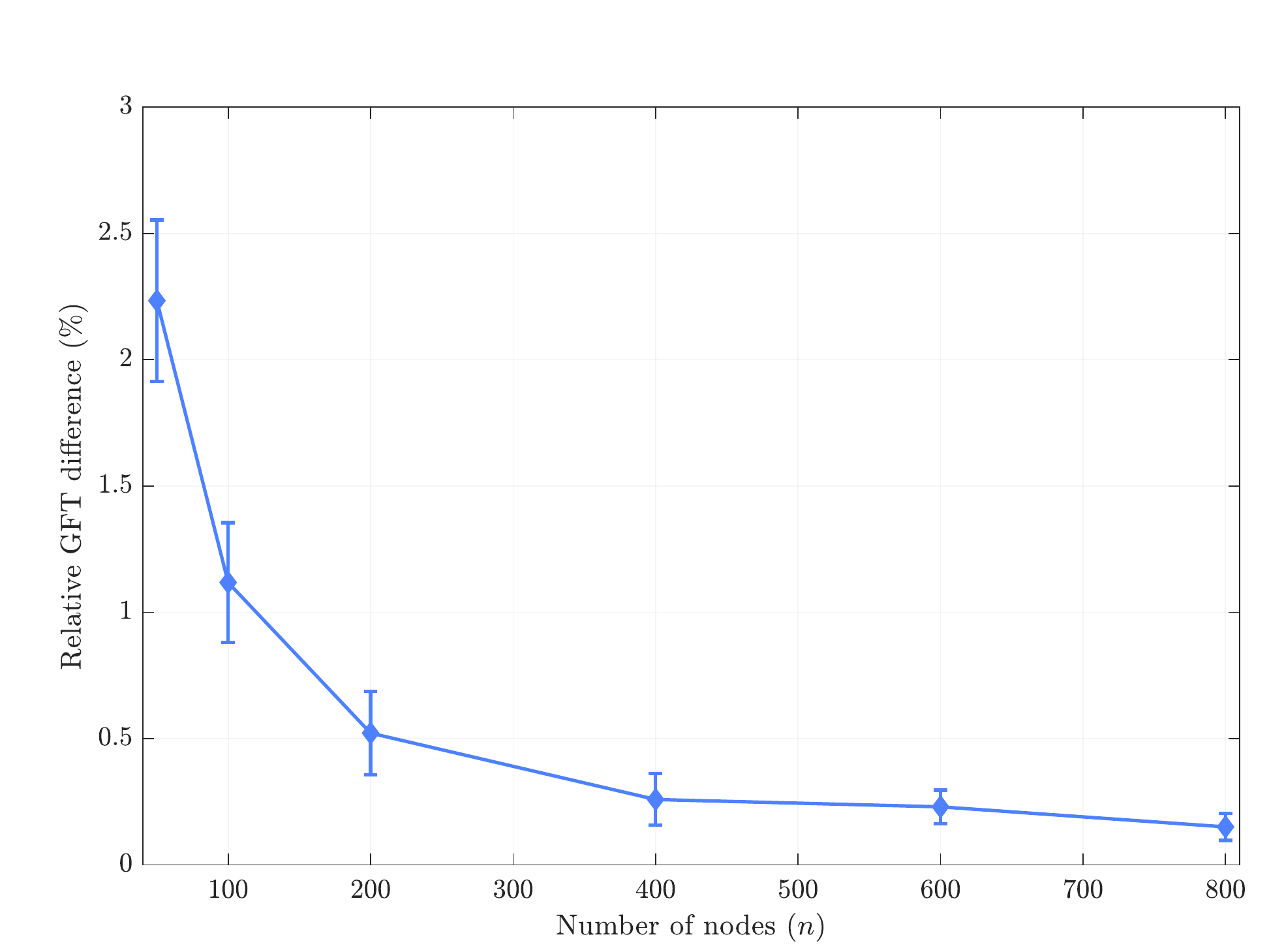}
  %\vspace{2.0cm}
\end{minipage}
\caption{Relative norm difference of the GFTs computed using the spectral basis of the $n$-node network and of the full 943-node network. Only the GFT coefficients associated with the first $n$ eigenvalues were considered.}
  \label{fig:movie_coeffs}
\end{figure}

%\begin{figure}[t]

%\begin{minipage}[t]{1.0\linewidth}
%  \centering
%  \centerline{\includegraphics[width=7cm]{val_training_loss-eps-converted-to.pdf}}
  %\vspace{2.0cm}
  %\centerline{(a) Result 1}\medskip
%\end{minipage}
%%
%%
%\begin{minipage}[b]{.48\linewidth}
%  \centering
%  \centerline{\includegraphics[width=4cm]{stevenson}}
  %\vspace{1.5cm}
%  \centerline{(b) Results 3}\medskip
%\end{minipage}
%\hfill
%\begin{minipage}[b]{0.48\linewidth}
%  \centering
%  \centerline{\includegraphics[width=4cm]{abbott}}  
  %\vspace{1.5cm}
%  \centerline{(c) Result 4}\medskip
%\end{minipage}
%%
%\caption{Training and validation losses in one round of simulation of the source localization problem on the Twitter graph, for the ReLu, 1-hop dynamic median and 2-hop dynamic median architectures .}
%\label{fig:loss}
%%
%\end{figure}

\section{Conclusions}
\label{sec:conclusions}

In this work, we have introduced a Fourier transform for graphon signals and shown that, for convergent sequences of graph signals, the GFT converges to the WFT. This result is empirically demonstrated in the numerical experiments of Section \ref{sec:sims}. The definition of graphon signals and of the WFT, together with Theorem \ref{thm:wft}, lay the ground for a graphon signal processing framework on which signals can be analyzed and systems designed in a \textit{centralized} and \textit{transferable} way, to then be applied to sequences of graphs or multiple instances of the same type of graph. 

\appendix

\section{Proof of Lemma 1}

The proof follows by direct computation. For $j \in L$, we have
\begin{align*}
	T_{\ccalW_G}&(\varphi_j)(u) = \int_0^1 \ccalW_G(u,v) \varphi_j(v) dv
	\\
	{}&= \sqrt{N} \mbI\left( u \in I_k \right) \int_0^1 [\bbS]_{k\ell} [\bbv_j]_{k} \times 
		\mbI\left( v \in I_\ell \right) dv
	\\
	{}&= \sqrt{N} \mbI\left( u \in I_k \right) \sum_{\ell = 1}^{N} [\bbS]_{k\ell} [\bbv_j]_{k} \int_{I_\ell} dv
	\\
	{}&= \frac{[\bbS \bbv_j]_{k}}{N} \times \sqrt{N} \mbI\left( u \in I_k \right)
		= \frac{\lambda_j(\bbS)}{N}
			\bigg[ [\bbv_j]_{k} \times \sqrt{N}\mbI\left( u \in I_k \right) \bigg]
	\\
	{}&= \sigma_j(T_{\ccalW_G}) \varphi_j(u)
		\text{.}
\end{align*}
If $j \notin L$, then $\langle{\varphi_j},{\varphi_k}\rangle = 0$ for all $k \in L$. In this case, we can trivially write $T_{\ccalW_G}(\varphi_j) = 0 = \sigma_j(T_{\ccalW_G}) \varphi_j(u)$. Note that since the $\bbv_k$ are orthonormal, so are the $\{\varphi_k(T_{\ccalW_G})\}$ and therefore a basis completion $\{\varphi_j\}$ can always be obtained. To conclude, compute for $j \in L$
\begin{align*}
	[\hat{X}]_j &= \int_0^1 \varphi_j(v) X(v) dv
	\\
	{}&= \sqrt{N} \int_0^1 [\bbv_j]_{\ell} [\bbx]_{\ell}
		\times \mbI\left( v \in I_\ell \right) dv
	\\
	{}&= \sqrt{N} \sum_{\ell = 1}^{N} [\bbv_j]_{\ell} [\bbx]_{\ell} \int_{I_\ell} dv
		= \dfrac{\bbv_j^\Tr \bbx}{\sqrt{N}} = \dfrac{[\hbx]_j}{\sqrt{N}}
		\text{.}
\end{align*}
If $j \notin L$, recall that since the $\{\bbv_j\}$ form a basis of $\reals^{N}$, we can write $\bbx = \sum_{k \in L} c_k \bbv_k$. Hence,
\begin{align*}
	[\hat{X}]_j &= \int_0^1 \varphi_j(v) X(v) dv
	\\
	{}&= \int_0^1 [\bbx]_{\ell} \times \mbI\left( v \in I_\ell \right)\varphi_j(v) dv
	\\
	{}&= \int_0^1 \sum_{k \in L} c_k [\bbv_k]_{\ell} \times \mbI\left( v \in I_\ell \right)\varphi_j(v) dv
	\\
	{}&= \dfrac{1}{\sqrt{N}} \sum_{k \in L} c_k \int_0^1 \varphi_k(v) \varphi_j(v) dv = 0
		\text{.} \qed
\end{align*}

\section{Proof of Lemma 2}

To prove Lemma \ref{T:eigenvectors_convergence}, we will need the following proposition that essentially shows that the spectrum of a convergent graph sequence tends to the spectrum of it limit graphon.

\begin{proposition}\label{T:eigenvalues_convergence}
Let $\{G_n\}$ be a sequence of graphs converging (in the sense of Definition \ref{defn:convergent_pairs}) to the non-derogatory graphon $\ccalW$ and denote by $\{\ccalW_n\}$ their induced graphons. Then, $\sigma_i(T_{\ccalW_n}) \to \sigma_i(T_\ccalW)$ for all $i \in \mbZ \setminus \{0\}$.
\end{proposition}

\begin{proof}
The proof is essentially the one for \cite[Theorem 6.7]{borgs2012convergent}, but we reproduce it here using our notation.

Recall that since the sequence $\{G_n\}$ converges to $\ccalW$, the density of homomorphisms for any finite graph also converges. The result then follows by choosing a homomorphism connected to the eigenvalues of their induced operators, namely the $k$-cycle $C_k$. Indeed, notice that for any graphon $\ccalW^\prime$ and $k \geq 2$, we have, by definition, that $t(C_k,\ccalW^\prime) = \sum_{i \in \mbZ \setminus \{0\}} \sigma_i(T_{\ccalW^\prime})^k$. Hence, we obtain that
\begin{equation}\label{E:k_cycle}
	\lim_{n \to \infty} \sum_{i \in \mbZ \setminus \{0\}} \sigma_i(T_{\ccalW_n})^k =
		\sum_{i \in \mbZ \setminus \{0\}} \sigma_i(T_\ccalW)^k
		\text{, for } k \geq 2
		\text{.}
\end{equation}
It now suffices to show that \eqref{E:k_cycle} implies $\sigma_i(T_{\ccalW_n}) \to \sigma_i(T_\ccalW)$.

We start by bounding the eigenvalues of any graphon $\ccalW^\prime$ in terms of its density of homomorphisms. In particular, for $k = 4$ we obtain that
\begin{align*}
	\sum_{i = 1}^m \sigma_i(T_{\ccalW^\prime})^4
		\leq \sum_{i \in \mbZ \setminus \{0\}} \sigma_i(T_{\ccalW^\prime})^4 = t(C_4,\ccalW^\prime)
	\Rightarrow
	\\
	\sigma_m(T_{\ccalW^\prime}) \leq \left[ \frac{t(C_4,\ccalW^\prime)}{m} \right]^{1/4}
		\text{ and}
	\\
	\sum_{i = -m}^{-1} \sigma_i(T_{\ccalW^\prime})^4
		\leq \sum_{i \in \mbZ \setminus \{0\}} \sigma_i(T_{\ccalW^\prime})^4 = t(C_4,\ccalW^\prime)
	\Rightarrow
	\\
	\sigma_{-m}(T_{\ccalW^\prime}) \geq -\left[ \frac{t(C_4,\ccalW^\prime)}{m} \right]^{1/4}
		\text{.}
\end{align*}
Since $t(C_4,\ccalW_n)$ is a convergent sequence, it has a bound $B$ \cite{borgs2012convergent}, which implies that
\begin{equation}\label{E:eigenvalue_bound}
	|{\sigma_i(T_{\ccalW^\prime})}| \leq \left( \frac{B}{|i|} \right)^{1/4}
		\text{, for all } i \in \mbZ \setminus \{0\}
		\text{.}
\end{equation}
Then, observe that for $k \geq 5$, we can take the limit in \eqref{E:k_cycle} term-by-term since $|{\sigma_i(T_{\ccalW_n})^k}| \leq (B/|i|)^{k/4}$ and the series $\sum_i (B/|i|)^{k/4}$ is convergent for $k > 4$. We therefore obtain from \eqref{E:k_cycle} that
\begin{equation}\label{E:limit_termbyterm}
	\lim_{n \to \infty} \sum_{i \in \mbZ \setminus \{0\}} \sigma_i(T_{\ccalW_n})^k
		= \sum_{i \in \mbZ \setminus \{0\}} \mu_i^k
		= \sum_{i \in \mbZ \setminus \{0\}} \sigma_i(T_\ccalW)^k
		\text{, for } k \geq 5
		\text{,}
\end{equation}
where $\mu_i^k = \lim_{n \to \infty} \sigma_i(T_{\ccalW_n})^k$.

To conclude, we proceed by induction over an ordering of the sequence of eigenvalues $\sigma_i(T_\ccalW)$, namely over $i_\ell$, $\ell = 1,2,\dots$, such that $|{\sigma_{i_1}(T_\ccalW)}| \geq |{\sigma_{i_2}(T_\ccalW)}| \geq \dots \geq |{\sigma_{i_\ell}(T_\ccalW)}|$. Suppose that $\mu_{i_\ell} = \sigma_{i_\ell}(T_\ccalW)$ for $\ell < \ell^*$ and let $\sigma_{i_{\ell^*}}(T_\ccalW)$ be of multiplicity $a$ and appear $b$ times in the sequence $\{\mu_i\}$ and $-\sigma_{i_{\ell^*}}(T_\ccalW)$ be of multiplicity $a^\prime$ and appear $b^\prime$ times in $\{\mu_i\}$. The identity in \eqref{E:limit_termbyterm} then reduces to
\begin{align}
	\left[ b + (-1)^k b^\prime \right]
		+ \sum_{\ell > \ell^*} \left( \frac{\mu_{i_{\ell}}}{\sigma_{i_{\ell^*}}(T_\ccalW)} \right)^k =
	\\
	\left[ a + (-1)^k a^\prime \right]
		+ \sum_{\ell > \ell^*} \left( \frac{\sigma_{i_{\ell}}(T_\ccalW)}{\sigma_{i_{\ell^*}}(T_\ccalW)} \right)^k
		\text{, for } k \geq 5
		\text{,}
\end{align}
where we divided both sides by $\sigma_{i_{\ell^*}}(T_\ccalW)^k$. Due to the ordering of the $\sigma_{i_\ell}$, for $k \to \infty$ through the even numbers we get $b + b^\prime = a + a^\prime$ and through the odd numbers we get $b - b^\prime = a - a^\prime$. Immediately, we have that $a = a^\prime$ and $b = b^\prime$, so that $\mu_{i_{\ell^*}} = \sigma_{i_{\ell^*}}$. Although this argument assumes $\mu_{i_\ell} < \sigma_{i_{\ell^*}}$ for all $\ell > \ell^*$, applying the same procedure to an ordering of the sequence $\{\mu_i\}$ yields the same conclusion.
\end{proof}

We will also require the following well known result about the perturbation of self-adjoint operators. For $\gamma$ a subset of the eigenvalues of a self-adjoint operator $T$, define the spectral projection~$E_T(\gamma)$ as the projection onto the subspace spanned by the eigenfunctions relative to those eigenvalues in $\gamma$. Then,

\begin{proposition}\label{T:davis_kahan}
Let $T$ and $T^\prime$ be two self-adjoint operators on a separable Hilbert space $\ccalH$ whose spectra are partitioned as $\gamma \cup \Gamma$ and $\omega \cup \Omega$ respectively, with $\gamma \cap \Gamma = \emptyset$ and $\omega \cap \Omega = \emptyset$. If there exists $d > 0$ such that $\min_{x \in \gamma,\, y \in \Omega} |{x - y}| \geq d$ and $\min_{x \in \omega,\, y \in \Gamma}|{x - y}| \geq d$, then
\begin{equation}\label{E:davis_kahan}
	\|E_T(\gamma) - E_{T^\prime}(\omega)\| \leq \frac{\pi}{2} \frac{\|{T - T^\prime}\|}{d}
\end{equation}
\end{proposition}

\begin{proof}
See \cite{seelmann2014notes}.
\end{proof}

Lastly, we also need two final results related to the graphon norm. The first, presented in Lemma \ref{cut_norm_conv}, states that if a sequence of graphs converges to a graphon in the homomorphism density sense, it also converges in cut norm, where the cut norm of a graphon $\ccalW: [0,1]^2 \to [0,1]$ is defined as \cite{lovasz2012large}
\begin{equation}
\|\ccalW\|_\square := \sup_{S,T \subseteq [0,1]} \bigg| \int_{S \times T} \ccalW(u,v)du dv \bigg|. 
\end{equation}

The second, here presented as Proposition \ref{T:norm_equivalence}, is due to \cite{lovasz2012large} and bounds the $L^2$ norm of the graphon operator by is cut norm.

\begin{lemma}[Cut norm convergence \cite{lovasz2012large}] \label{cut_norm_conv}
If $\{G_n\}$ is a sequence of graphs converging to the graphon $\ccalW$ in the homomorphism density sense, then the there exists a sequence of permutations $\{\pi_n\}$ such that
\begin{equation}
\|\ccalW_{\pi_n(G_n)} - \ccalW\|_\square \to 0
\end{equation}
where $\ccalW_{G_n}$ is the graphon induced by the graph $G_n$.
\end{lemma}
\begin{proof}
See \cite[Theorem 11.57]{lovasz2012large}.
\end{proof}

\begin{proposition}\label{T:norm_equivalence}
Let $T_\ccalW$ be the operator induced by the graphon $\ccalW$. Then, $\|{\ccalW}\|_\square \leq \|{T_\ccalW}\| \leq \sqrt{8 \|{\ccalW}\|_\square}$.
\end{proposition}

This is a direct consequence of \cite[Theorem 3.7(a)]{borgs2008convergent} and the fact that $t(C_2,\ccalW)$ is the Hilbert-Schmidt norm of $T_\ccalW$, which dominates the $L^2$-induced operator norm.

We can now proceed with the proof of our lemma:

\begin{proof}[Proof of Lemma \ref{T:eigenvectors_convergence}]
For~$j \in C$, let~$\gamma = \sigma_j(T_\ccalW)$, $\Gamma = \{\sigma_i(T_\ccalW)\}_{i \neq j}$, $\omega = \sigma_j(T_{\ccalW_n})$, and $\Omega = \{\sigma_i(T_{\ccalW_n})\}_{i \neq j}$ in Proposition \ref{T:davis_kahan} to get
\begin{equation}\label{E:davis_kahan_1}
	\|{E_j - E_{jn}}\| \leq \frac{\pi}{2} \frac{\|{T_{\ccalW_n} - T_\ccalW}\|}{d_{jn}}
\end{equation}
where $E_j$ and $E_{jn}$ are the spectral projections of $T_\ccalW$ and $T_{\ccalW_n}$ with respect to their $j$-th eigenvalue and
\begin{multline*}
	d_{jn} = \min \big( |{\sigma_j - \sigma_{j+1}(T_{\ccalW_n})}|,
		|{\sigma_j - \sigma_{j-1}(T_{\ccalW_n})}|, \\
		|{\sigma_{j+1} - \sigma_{j}(T_{\ccalW_n})}|,
		|{\sigma_{j-1} - \sigma_{j}(T_{\ccalW_n})}|
	\big)
		\text{,}
\end{multline*}
where we omitted the dependence on $\ccalW$ by writing $\sigma_j = \sigma_j(T_\ccalW)$.

Fix $\epsilon > 0$. From Proposition \ref{T:eigenvalues_convergence}, we know we can find $N_1$ such that $|{d_{jn} - \delta_j}| \leq \delta_j/2$ for all $n > N_1$, where 
\begin{equation*}
\delta_j = \min\big( |{\sigma_j - \sigma_{j+1}}|, |{\sigma_j - \sigma_{j-1}}| \big)\ .
\end{equation*} 
Since $\ccalW$ is non-derogatory, $\delta_j > 0$. Additionally, the cut norm convergence of graphon sequences (Lemma \ref{cut_norm_conv}) together with Proposition \ref{T:norm_equivalence} implies there exists $N_2$ such that $\|{T_{\ccalW_n} - T_\ccalW}\| \leq \epsilon \delta_j/\pi$. Hence, for all $n > \max(N_1,N_2)$ it holds from \eqref{E:davis_kahan_1} that
\begin{equation}\label{E:convergence_calA}
	\|{E_j - E_{jn}}\| \leq \frac{\pi}{2} \frac{\epsilon \delta_j/\pi}{\delta_j/2} = \epsilon
		\text{.}
\end{equation}
Since $\epsilon$ is arbitrary, \eqref{E:convergence_calA} proves that the projections onto the eigenfunctions of the same eigenvalue converge. In other words, the eigenfunction sequence $\varphi_j(T_{\ccalW_n})$ itself converges weakly.

To proceed, let us apply Proposition \ref{T:davis_kahan} to the subspace spanned by the remaining eigenfunctions with indices not in $C$. Explicitly, let $\gamma = \{\sigma_i(T_\ccalW)\}_{i \notin C}$, $\Gamma = \{\sigma_i(T_\ccalW)\}_{i \in C}$, $\omega = \{\sigma_i(T_{\ccalW_n})\}_{i \notin C}$, and~$\Omega = \{\sigma_i(T_{\ccalW_n})\}_{i \in C}$ in \eqref{E:davis_kahan} to get
\begin{equation}\label{E:davis_kahan_2}
	\|{E^\prime - E^\prime_{n}}\| \leq \frac{\pi}{2} \frac{\|{T_{\ccalW_n} - T_\ccalW}\|}{d_{n}}
		\text{,}
\end{equation}
where $E^\prime$ and $E^\prime_n$ are the projections onto the subspaces given by $\ccalE = \text{span}\left( \{\varphi_i(T_\ccalW)\}_{i \notin C} \right)$ and $\ccalE_n = \text{span}\left( \{\varphi_i(T_{\ccalW_n})\}_{i \notin C} \right)$ respectively and $d_{n} = \inf_{i \notin C} \left( |{c - \sigma_{i}(T_{\ccalW_n})}| \right)$. Since the graphon $\ccalW$ is non-derogatory, there exists an $N_0$ such that $d_n > 0$ for all $n > N_0$ and we can use the same argument as above to obtain that $E^\prime_{n} \to E^\prime$ in operator norm.

To see how this implies that $\varphi_i(T_{\ccalW_n}) \to \phi \in \ccalE$ for all $i \notin C$, suppose this is not the case. Then, $\|{\phi - E^\prime(\phi)}\| \geq \epsilon > 0$ since $\phi \notin \ccalE$. Without loss of generality, we assume that $\|{\phi}\| = 1$ (if not, simply normalize $\phi$: since $\ccalE$ is a subspace $\phi \notin \ccalE \Leftrightarrow K \phi \notin \ccalE$ for any $K > 0$). Notice, however, that there exists $N^\prime$ such that $\|{\varphi_i(T_{\ccalW_n}) - \phi}\| \leq \epsilon/8$ and $\|{E^\prime(\phi) - E_n^\prime(\phi)}\| \leq \epsilon/4$ for all $n > N^\prime$, which implies that $\|{\phi - E^\prime(\phi)}\| \leq \epsilon/2$, contradicting the hypothesis. Indeed,
\begin{align*}
	\|{\phi - E^\prime(\phi)}\| =& \|\phi - \varphi_i(T_{\ccalW_n})
		+ E_n^\prime(\phi) - E^\prime(\phi) + 
		\\		
		&E_n^\prime(\varphi_i(T_{\ccalW_n}) - \phi)\| \leq \|{\phi - 	\varphi_i(T_{\ccalW_n})}\| + 
		\\		
		&\|{E_n^\prime(\phi) - E^\prime(\phi)}\| + \|{E_n^\prime(\varphi_i(T_{\ccalW_n}) - \phi)}\|
		\text{.}
\end{align*}
Then, using Cauchy-Schwarz and the fact that $E_n^\prime$ is an orthogonal projection, i.e., $\|{E_n^\prime}\| = 1$, yields
\begin{align*}
	\|{\phi - E^\prime(\phi)}\| \leq 2\|{\phi - \varphi_i(T_{\ccalW_n})}\|
		+ \|{E_n^\prime(\phi) - E^\prime(\phi)}\|
		\text{.}
\end{align*}
which for all $n > N^\prime$ reduces to
\begin{align*}
	\|{\phi - E^\prime(\phi)}\| &\leq \frac{\epsilon}{2}
		\text{,}
\end{align*}
contradicting the fact that~$\phi \notin \ccalE$ and concluding the proof.
\end{proof}

\section{The space of non-derogatory graphons is dense}
\begin{proposition} \label{T:non_derog_graphons_dense}
Non-derogatory graphons are dense in the space of graphons with respect to the cut norm.
\end{proposition}
\begin{proof}
This is due to the fact that the operators induced by non-derogatory graphons are dense in the topology induced by the $L^2$ operator norm on the space of compact, self-adjoint operators, cf. Proposition \ref{T:nonderogatory_density} below. Since this topology is equivalent to the one induced by the cut norm, this implies that non-derogatory graphons are also dense in the space of graphons with respect to the cut norm.
\end{proof}

\begin{proposition}\label{T:nonderogatory_density}
The set of operators induced by non-derogatory graphons is dense in the space of linear, compact, self-adjoint operators with respect to the $L^2$-induced norm.
\end{proposition}

\begin{proof}

This is a direct consequence of the fact that every compact, self-adjoint operator is the limit of a sequence of finite rank operators. To see why this is the case, recall that the eigenfunctions $\{\varphi_i\}$ form an orthonormal basis of $L^2([0,1])$ \cite[Chapter 28, Theorem 3]{lax02-functional}. Hence, since $\ccalW \in L^2([0,1]^2)$, the induced $T_\ccalW$ has finite $L^2$-norm and the sequence $\sum_{i \in \mbZ \setminus \{0\}} |{\langle{T_{\ccalW}(X)},{\varphi_i}\rangle}|^2$ is convergent and can be arranged so that for every $\epsilon > 0$, there exists $N$ such that
\begin{equation}\label{E:small_tail}
	\sum_{|{i}| > n} |{\langle{T_{\ccalW}(X)},{\varphi_i}\rangle}|^2 \leq \frac{\epsilon^2 \|{X}\|}{2}
		\text{, for all } n > N \text{.}
\end{equation}

Fix a graphon $\ccalW$. We now show that for any $\epsilon > 0$, there exists a non-derogatory graphon $\ccalW^\prime$ such that $\|{T_\ccalW - T_{\ccalW^\prime}}\| \leq \epsilon$. To do so, define the graphon $\ccalW_n$ through its operator as in
\begin{equation}
	T_{\ccalW_n}(X) = \sum_{|{i}| \leq n}
		\langle{T_{\ccalW}(X)},{\varphi_i}\rangle \varphi_i
		+ \sum_{|{i}| \leq n} \delta_i \varphi_i
		\text{,}
\end{equation}
where the $\delta_i$ are chosen so that $\sigma_i + \delta_i \neq \sigma_j + \delta_j$ for all $|{i}|,|{j}| \leq n$ and $|{\delta_i}| \leq \epsilon/(2\sqrt{n})$. In other words, the $\delta_i$ are small perturbations chosen to guarantee that $T_{\ccalW_n}$ is non-derogatory. Since the $\{\varphi_i\}$ form an orthonormal basis, we obtain that
\begin{align*}
	\|{T_\ccalW - T_{\ccalW_n}}\|^2 &= \sup_{\|{X}\| = 1} \|{T_\ccalW(X) - T_{\ccalW_n}(X)}\|^2	=
	\\
	{}&\sum_{|{i}| \leq n} \delta_i^2
		+ \sup_{\|{X}\| = 1} \sum_{|{i}| > n} |{\langle{T_{\ccalW}(X)},{\varphi_i}\rangle}|^2
	\\
	{}&\leq \frac{\epsilon^2}{2}
		+ \sup_{\|{X}\| = 1} \sum_{|{i}| > n} |{\langle{T_{\ccalW}(X)},{\varphi_i}\rangle}|^2
		\text{.}
\end{align*}
Using \eqref{E:small_tail} and taking $\ccalW^\prime = \ccalW_N$, we conclude that $\|{T_\ccalW - T_{\ccalW^\prime}}\| \leq \epsilon$.
\end{proof}

\bibliographystyle{IEEEbib}
\bibliography{myIEEEabrv,bib-graphon}

\end{document}